\documentclass[oneside,a4paper,12pt]{amsart}%
\usepackage{amsfonts}
\usepackage{amsmath}
\usepackage{esint}
\usepackage{times}
\usepackage{hyperref}
\usepackage{dsfont}
\usepackage{graphicx, subfigure}
\usepackage{color}
\usepackage{bm}
\usepackage{enumerate}
\usepackage{amssymb}
\usepackage{graphicx}%
\providecommand{\U}[1]{\protect\rule{.1in}{.1in}}
\numberwithin{equation}{section}
\providecommand{\U}[1]{\protect\rule{.1in}{.1in}}
\newtheorem{theorem}{Theorem}

\newtheorem{assumption}[theorem]{Assumption}

\theoremstyle{definition}
\newtheorem{example}[theorem]{Example}

\newtheorem{lemma}[theorem]{Lemma}

\newtheorem{proposition}[theorem]{Proposition}
\newtheorem{remark}[theorem]{Remark}

\begin{document}
\title{Negative potentials and collapsing universes II}
\author{Roberto Giamb\`{o}}
\address{School of Science and Technology, Mathematics Division, University of
Camerino, 62032 Camerino (MC) Italy.}
\email{roberto.giambo@unicam.it}
\author{John Miritzis}
\email{imyr@aegean.gr}
\author{Koralia Tzanni}
\address{Department of Marine Sciences, University of the Aegean, University Hill,
Mytilene 81100, Greece.}
\email{tzanni@aegean.gr}

\begin{abstract}
Completing a previous analysis started in \cite{gmt}, we study flat
Friedmann--Lema\^{\i}tre--Robertson--Walker (FLRW) models with a perfect fluid
matter source and a scalar field nonminimally coupled to matter,
self--interacting with a potential that may attain negative values. We prove
that the evolution generically forces the Hubble function to diverge to
$-\infty$ in a finite time, except in case the potential exhibits a flat
plateau at infinity (tending to zero from below); in that case we find
conditions which may give rise to ever expanding or recollapsing cosmologies.

\end{abstract}
\maketitle

\section{Introduction}

An accelerating phase of the Universe requires scalar fields with non-negative
potentials playing the role of a cosmological term. However, in some
cosmological models potentials taking negative values are used, (see
\cite{ffkl} for motivations). In these cases several authors conclude that the
Universe eventually collapses even if it is flat, \cite{ffkl}-\cite{kowa}. In
a recent paper we rigorously proved that a general class of bounded from above
potentials with $\lim_{\phi\rightarrow-\infty}V\left(  \phi\right)  =-\infty$,
almost always forces the Hubble function $H$ to diverge to $-\infty$ in a
finite time, \cite{gmt}. This means that, up to a non generical choice of
initial data, an initially expanding Universe recollapses and develops a
singularity in a finite amount of time. The problem of collapsing scalar
fields cosmologies from the mathematical point of view has been sporadically
treated in the literature, e.g. \cite{giam1}-\cite{rzmm}. On the other hand,
non-negative potentials in FLRW models with mathematically rigorous results,
have been studied by several authors (see for example
\cite{foster,rend,gimi,ahu}), but to our knowledge there is not a
corresponding rigorous treatment of negative potentials, apart from the above
mentioned study in \cite{gmt}.

The class considered in \cite{gmt} contains potentials that fall to $-\infty$
as $\phi\rightarrow-\infty$, have a global positive maximum and go to zero
from above as $\phi\rightarrow+\infty$. The purpose of the present paper is to
complete the rigorous treatment of cosmological models with potentials taking
negative values. The remaining forms of negative potentials encountered in the
literature are listed as follows.

\begin{enumerate}
[A.]

\item \label{itm:A} Potentials having a negative minimum. Two important
examples include the ekpyrotic potentials and those used in models of cyclic
Universes; for reviews see Refs. \cite{linde2}, \cite{lehne2}.

\item \label{itm:B} Bounded from below potentials with no minimum. As an
example, we mention the potentials
\[
V\left(  \phi\right)  =V_{0}e^{-\lambda\phi}-C,\ \ \ V_{0},C,\lambda>0,
\]
which were considered in the context of supersymmetry theories, see for
example Ref. \cite{kali}.

\item \label{itm:C} Potentials with $V\left(  \phi\right)  $ decreasing from
$+\infty$ to $-\infty,$ for example
\[
V\left(  \phi\right)  =W_{0}-V_{0}\sinh\left(  \lambda\phi\right)
,\ \ \ \ \ \lambda,V_{0}>0,
\]
(see \cite{gma} where an exact solution was obtained in the absence of matter).

\item \label{itm:D} Potentials having a global positive maximum and
$\lim_{\phi\rightarrow\pm\infty}V\left(  \phi\right)  =-\infty.$ Near the
maximum, say at $\phi=0,$ they can be represented as%
\[
V\left(  \phi\right)  =V_{0}-\frac{m^{2}}{2}\phi^{2},
\]
cf. \cite{ffkl}. An example is the potential%
\[
V\left(  \phi\right)  =V_{0}\left(  2-\cosh\left(  \sqrt{2}\phi\right)
\right)  ,\ \ \ V_{0}>0,
\]
considered in \cite{kali}. Potentials of this class appear in cosmological
models in $N=2,4,8$ gauged supergravity \cite{linde1,klps1}. For detailed
cosmological implications see \cite{klps,kali,kali1}.
\end{enumerate}

From the mathematical point of view, potentials of type \ref{itm:A},
\ref{itm:B} and \ref{itm:C} cannot be studied using the techniques exploited
in \cite{gmt}; on the other hand, for potentials of type \ref{itm:D} similar
arguments as those used in \cite{gmt} apply almost straightforwardly.

In this paper we study the recollapse problem of scalar-field cosmological
models with negative potentials having the general features of the above
classes \ref{itm:A}--\ref{itm:D}. We consider a nonminimal coupling of the
scalar field to matter; the coupling coefficient is assumed to be an arbitrary
non-negative bounded function of the scalar field with non-negative limits as
$\phi\rightarrow\pm\infty$, see (\ref{eq:Q}) below. Under quite general
assumptions on the potential, initially expanding flat Universes are shown to
eventually recollapse, except in case where the potential exhibits a flat
plateau at infinity, tending to 0 from below. In this case, generical
evolution may also be eternally expanding for some open subsets in the
parameter space.

The plan of the paper is as follows. In the next section we write the field
equations for flat FLRW models as a constrained four-dimensional dynamical
system. We impose a number of assumptions so that our class of potentials
includes cases \ref{itm:A}, \ref{itm:B} and \ref{itm:C}. We state our main
theorem according to which, the Hubble function, $H$, almost always diverges
to $-\infty$ in a finite time. In section \ref{sec3} we analyse all possible
limit sets of the dynamical system and prove a number of propositions that
lead to the proof of the main theorem. Section \ref{sec4} contains a
proposition covering case \ref{itm:D}, and so we complete our analysis that is
discussed in the final section.

\section{Formulation of the problem and main theorem}

For homogeneous and isotropic flat spacetimes the field equations can be
written (see \cite{gmt}), as an autonomous dynamical system,
\begin{align}
\dot{\phi}  &  =y,\label{eq:1}\\
\dot{y}  &  =-3Hy-V^{\prime}(\phi)+\alpha\rho,\qquad\alpha=\frac{4-3\gamma}%
{2}Q(\phi),\label{eq:2}\\
\dot{\rho}  &  =-\rho(3\gamma H+\alpha y),\label{eq:3}\\
\dot{H}  &  =-\frac{1}{2}\left(  y^{2}+\gamma\rho\right)  , \label{eq:4}%
\end{align}
subject to the constraint,
\begin{equation}
3H^{2}=\frac{1}{2}y^{2}+V(\phi)+\rho. \label{eq:5}%
\end{equation}
In most quintessence models, the coupling coefficient, $Q$, is postulated to
be a positive constant, see for example \cite{amen1}; here $Q$ is assumed to
be a positive and bounded function of class $C^{1}$ such that,
\begin{equation}
Q_{\pm}:=\lim_{\phi\rightarrow\pm\infty}Q(\phi)>0. \label{eq:Q}%
\end{equation}
For motivation and other couplings see \cite{leon,tzmi}. We recall that the
constraint \eqref{eq:5} is invariant under the flow of
\eqref{eq:1}--\eqref{eq:4}. In the following we will consider solutions to
\eqref{eq:1}--\eqref{eq:4} which satisfy \eqref{eq:5} at some initial time,
and therefore \eqref{eq:5} holds throughout the evolution. We will refer to
this solution as solutions to \eqref{eq:1}--\eqref{eq:5}.

In the following we incorporate cases \ref{itm:A}--\ref{itm:C} into a large
class of potentials $V(\phi)\in\mathcal{C}^{2}$ satisfying some further
assumptions. To begin, let $u(\phi)$ be the function,
\begin{equation}
u(\phi)=\frac{V^{\prime}(\phi)}{V(\phi)}. \label{eq:u}%
\end{equation}

\begin{assumption}
\label{as:1} We assume that $V(\phi)\in\mathcal{C}^{2}$ is such that

\begin{enumerate}
\item $\lim_{\phi\to-\infty}V(\phi)=+\infty$,

\item There exists a unique $\phi_{0}\in\mathbb{R\,}:\,V(\phi_{0})=0$.
Moreover, $V$ is strictly decreasing for all $\phi\leq\phi_{0}$,

\item $\lim_{\phi\rightarrow\pm\infty}u(\phi)=\lambda_{\pm}\in\mathbb{R}$,

\item $\lim_{\phi\rightarrow+\infty}V(\phi)=V_{\infty}\leq0$ (possibly
$V_{\infty}=-\infty$).

\item \label{itm:5}There exists a $\mathcal{C}^{2}$--diffeomorphism,
$f:(-\infty,\phi_{0}]\rightarrow\lbrack0,s_{0})$, such that

\begin{enumerate}
\item \label{itm:a}The limit $\lim_{\phi\rightarrow-\infty}f^{\prime}(\phi
)$,$\ $exists and is equal to zero,

\item $\lim_{\phi\rightarrow-\infty}f(\phi)=0$,

\item $\lim_{\phi\rightarrow-\infty}\frac{u^{\prime}(\phi)}{f^{\prime}(\phi
)}=0$,

\item \label{itm:d} $\lim_{\phi\rightarrow-\infty}\frac{f^{\prime\prime}%
(\phi)}{f^{\prime}(\phi)}\in\mathbb{R\,}$.
\end{enumerate}

\item \label{itm:6} If $V_{\infty}=0$, then there exists a $\phi_{M}>0$, such
that $V$ is strictly increasing for $\phi\geq\phi_{M}$. Moreover, we make a
similar hypothesis to \eqref{itm:5} above for $\phi\rightarrow+\infty$,
assuming the existence of a $\mathcal{C}^{2}$--diffeomorphism, $g(\phi
):[\phi_{M},+\infty)\rightarrow(0,s_{0}]$ such that requests
\eqref{itm:a}--\eqref{itm:d} hold for $g(\phi)$, as $\phi\rightarrow+\infty$.
\end{enumerate}
\end{assumption}

Assumptions \eqref{itm:5}--(\ref{itm:6}) are required for situations where the
scalar field possibly diverges. In those cases, the above diffeomorphisms are
needed to bring a neighbourhood of infinity to a neighbourhood of the origin,
\cite{foster,giam}.

It is easy to verify that cases \ref{itm:A}, \ref{itm:B} and \ref{itm:C} of
negative potentials mentioned in the Introduction satisfy Assumption
\ref{as:1}. In the next section we will examine the possible $\omega-$limit
sets of the system, essentially depending on the asymptotic behavior of the
scalar field $\phi(t)$. We will see that, except one case described in
Proposition \ref{thm:propV0}, solutions to \eqref{eq:1}--\eqref{eq:5} always
recollapse to a singularity in a finite amount of time. In particular, the
results proved in the next section may be collected in the following main theorem.

\begin{theorem}
\label{thm:main} Let $V(\phi)$ satisfy Assumption \ref{as:1}. Then, if either
$V_{\infty}<0$, or condition,
\begin{equation}
0<\gamma<\frac{4}{3},\quad0<Q_{+}<\sqrt{6}\frac{2-\gamma}{4-3\gamma}%
,\quad\lambda_{+}<-\frac{4-3\gamma}{2}Q_{+}-\frac{3(2-\gamma)\gamma
}{(4-3\gamma)Q_{+}}\label{eq:cond}%
\end{equation}
does \textbf{not} hold, then a solution to \eqref{as:1}--\eqref{eq:5}, up to a
zero--measured set of initial data, recollapses to a singularity in a finite
amount of time, i.e.,
\begin{equation}
\exists t_{\ast}>0:\,\lim_{t\rightarrow t_{\ast}^{-}}H(t)=-\infty
.\label{eq:rec}%
\end{equation}
Otherwise, if $V_{\infty}=0$ and \eqref{eq:cond} does hold, a solution to
\eqref{eq:1}--\eqref{eq:5} either generically recollapses to a singularity in
a finite time or expands forever, with $\phi(t)\rightarrow+\infty$ and
$y(t),\rho(t)$ and $H(t)$ infinitesimal as $t\rightarrow+\infty$.
\end{theorem}

\section{Qualitative behavior of the solution\label{sec3}}

Throughout this section we will suppose that $V(\phi)$ satisfies the set of
hypotheses collected in Assumption \ref{as:1}. We will call $\phi_{\infty}%
\in\mathbb{R}\cup\{\pm\infty\}$ the limit value of $\phi(t)$, if it exists,
i.e. $\lim_{t\rightarrow\sup\mathbb{I}}\phi(t)=\phi_{\infty}$, where
$\mathbb{I}$ is the maximal interval of definition of a solution to
\eqref{eq:1}--\eqref{eq:5}. Some of the proofs in this section rely on two
lemmas proved in \cite{gmt}; for the convenience of the reader we reproduce
them here.

\begin{lemma}
\label{lem:1}Let $\gamma(t)=\left(  \phi(t),y(t),\rho(t),H(t)\right)  $ be a
bounded solution such that $\rho(t_{0})>0$. Then $\gamma(t)\in W^{s}(q_{\pm}%
)$, where $W^{s}(q)$ is the stable manifold of an equilibrium point $q$.
\end{lemma}

It can be shown that the dimension of the stable manifold is always less than
the dimension of the phase space, \cite{gmt}. Therefore, the meaning of the
above lemma is that, \emph{future bounded trajectories of the system are not
generic}.

\begin{lemma}
\label{lem:2} Let $\gamma(t)$ be a solution to the system
\eqref{eq:1}--\eqref{eq:5}. If there exists $t_{1}\geq t_{0}$ and $\bar{V}%
\in\mathbb{R}$ such that, for all $t\geq t_{1}$, $V(\phi(t))\leq\bar{V}$, and
either $(\imath)$ $\bar{V}<0$, or $(\imath\imath)$ $H(t_{1})<-\sqrt{\bar{V}%
/3}$, then $H(t)$ negatively diverges in a finite time, i.e. (\ref{eq:rec}) holds.
\end{lemma}

We are now ready to examine different situations depending on $\phi_{\infty}$.

\subsection{Case $\phi_{\infty}=-\infty$}

We firstly analyse the case \ref{itm:C} when the scalar field negatively
diverges in such a way that $V(\phi(t))\rightarrow+\infty$, see for example
\cite{gma}. We use expansion normalized variables techniques, first introduced
in \cite{clw}; see also \cite{foster,giam}. Notice that, in the present paper
we do not need to assume an a priori estimate on $\lambda_{\pm}$, unlike for
instance in \cite[eq.(2.2)]{giam}.

\begin{proposition}
\label{thm:prop2} If $\phi_{\infty}=-\infty$ then \eqref{eq:rec} generically holds.
\end{proposition}

\begin{proof}
Since $V(\phi(t))\rightarrow+\infty$, then $H(t)^{2}\rightarrow+\infty$, and
since $H(t)$ is decreasing we conclude that $H\left(  t\right)  \ $negatively
diverges, so we have to prove that this happens in a finite amount of time.
Without loss of genericity suppose $H(0)<0$. We introduce the variables,
\[
\phi,\qquad x=\frac{1}{H},\qquad w=\frac{y}{H},\qquad z=\frac{\sqrt{\rho}}%
{H},
\]
and a new time coordinate $\tau$, defined by $\mathrm{d}\tau/\mathrm{d}t=-H$,
as done in \cite{giam}. Using the constraint
\[
V(\phi)x^{2}+\frac{1}{2}w^{2}+z^{2}=3,
\]
to eliminate $x(\tau)$, we come to the following system for the triple
$(\phi(\tau),w(\tau),z(\tau))$:
\begin{align}
\frac{\mathrm{d}\phi}{\mathrm{d}\tau}  &  =-w,\label{eq:A1}\\
\frac{\mathrm{d}w}{\mathrm{d}\tau}  &  =-\left(  \frac{1}{2}w^{2}-3\right)
(w+u(\phi))-z^{2}\left(  \frac{\gamma}{2}w+\alpha(\phi)+u(\phi)\right)
,\label{eq:A2}\\
\frac{\mathrm{d}z}{\mathrm{d}\tau}  &  =-\frac{1}{2}z\left[  w^{2}-\alpha
w+\gamma(z^{2}-3)\right]  , \label{eq:A3}%
\end{align}
where $u(\phi)=V^{\prime}(\phi)/V\left(  \phi\right)  $ was already defined.
We recall that we are interested in the dynamics near the critical point
\textquotedblleft at infinity\textquotedblright,\ $\phi\rightarrow-\infty$.
Therefore, we introduce the variable $s=f(\phi)$, where $f$ is defined in
Assumption \ref{as:1}. In this way we obtain a system in the variables
$(w(\tau),z(\tau),s(\tau))$, ruled by equations \eqref{eq:A2}, \eqref{eq:A3}
and
\begin{equation}
\frac{\mathrm{d}s}{\mathrm{d}\tau}=-wf^{\prime-1}(s)). \label{eq:A4}%
\end{equation}
Remembering that, $V(\phi(t))$ is eventually positive, we are interested in
solutions to \eqref{eq:A2}--\eqref{eq:A4} such that,
\begin{equation}
\frac{1}{2}w^{2}+z^{2}<3. \label{eq:A5}%
\end{equation}
We consider critical points of \eqref{eq:A2}--\eqref{eq:A4} such that $s=0$,
which are candidates to be $\omega$--limit points for the solutions we are
interested in. We may further restrict ourselves to critical points with
$w\geq0$, since we expect both $y$ and $H$ to be eventually negative and
$z\leq0$. The $(w,z)$--coordinates of the admissible critical points are then
(setting $\lambda=\lambda_{-}$ and $\alpha=\alpha(f^{-1}(0))=\tfrac{4-3\gamma
}{2}Q_{-}$)
\begin{align*}
\mathcal{A}  &  =\left(  \sqrt{6},0\right)  ,\,\mathcal{B}=(-\lambda
,0),\,\mathcal{C}=\left(  \frac{2\alpha}{2-\gamma},\frac{\sqrt{-2\alpha
^{2}+3(2-\gamma)^{2}}}{\gamma-2}\right)  ,\,\\
\mathcal{D}  &  =\left(  -\frac{3\gamma}{\alpha+\lambda},\frac{\sqrt
{3(-3\gamma+\alpha\lambda+\lambda^{2})}}{\alpha+\lambda}\right)  .
\end{align*}
It is easy to check that all these points, except possibly $\mathcal{B}$, do
not coincide with the origin $(0,0)$. In the particular case when $\lambda=0$,
then $\mathcal{B}=(0,0)$, but the eigenvalues of the linearised system
associated with this critical point are $\{0,3,\frac{3}{2}\gamma\}$ and so
this point is definitely an unstable equilibrium.

The generical situation therefore is that there exists a $\delta
>0\,:\,\frac{1}{2}w^{2}+z^{2}\geq\delta$ eventually. Then, recalling
\eqref{eq:4},
\[
\frac{1}{H(t)}-\frac{1}{H(t_{0})}=\int_{t_{0}}^{t}-\frac{\dot{H}(\sigma
)}{H(\sigma)^{2}}\,\mathrm{d}\sigma=\int_{t_{0}}^{t}\frac{1}{2}(w^{2}+\gamma
z^{2})\,\mathrm{d}\sigma\geq\delta(t-t_{0}).
\]
Since $H(t_{0})<0$, we conclude that $H(t)$ diverges in a finite amount of time.
\end{proof}

\begin{remark}
\label{rem:Vdiv} The same dynamics described in the above proposition applies
to the more general case $\liminf_{t\rightarrow\sup\mathbb{I}}V(\phi
(t))=-\infty$. Indeed, this situation implies again that $H(t)\rightarrow
-\infty$, and in the above proposition we have proved that the dynamics near
the point \textquotedblleft at infinity\textquotedblright\ $\phi
\rightarrow-\infty$, give necessarily rise to solutions that completely
recollapse in a finite time.
\end{remark}

\subsection{Case $\phi_{\infty}\in\mathbb{R}$.}

We briefly examine what happens if the scalar field converges to a positive value.

\begin{proposition}
\label{thm:propphifin} If $\phi_{\infty}\in\mathbb{R}$ then \eqref{eq:rec}
generically holds.
\end{proposition}

\begin{proof}
If $\phi_{\infty}\in\mathbb{R}$ then $\phi(t)$ is bounded. Then, if $H(t)$ was
bounded too, by \eqref{eq:5} also $y(t),\rho(t)$ would be bounded so the
solution would be bounded which by Lemma \ref{lem:1} is a non generic
situation. Then $H(t)$ is unbounded and since it is decreasing by
\eqref{eq:4}, it must diverge to $-\infty$. At this point, Lemma \ref{lem:2}
applies to give the result.
\end{proof}

\subsection{Case $\phi_{\infty}=+\infty$}

In this situation we must split the argument into two subcases, depending on
the value of $V_{\infty}$. We start by considering shortly the case when this
limit value is strictly negative (case \ref{itm:C}), possibly $-\infty$ (case
\ref{itm:D}).

\begin{proposition}
\label{thm:propphiinftyVneg} If $V_{\infty}<0$ and $\phi_{\infty}=+\infty$
then \eqref{eq:rec} generically holds.
\end{proposition}

\begin{proof}
If $\phi_{\infty}=+\infty$ then, since $V_{\infty}<0$, there exists a $\bar
{V}<0$, such that, $V(\phi(t))\leq\bar{V}<0$ eventually and then Lemma
\ref{lem:2} applies again to give the result.
\end{proof}

A more subtle case happens when $V_{\infty}=0$, as is the case of the
ekpyrotic potentials. In this situation the critical point \textquotedblleft
at infinity\textquotedblright\ corresponding to $\phi\rightarrow+\infty$ must
be studied carefully, since it may give rise to ever expanding cosmologies.
Before we state the precise theorem, the following preliminary result is needed.

\begin{lemma}
\label{thm:lem} If $V_{\infty}=0$ and $\phi_{\infty}=+\infty$, then
$H_{\infty}=\lim_{t\rightarrow\sup\mathbb{I}}H(t)\leq0$.
\end{lemma}

\begin{proof}
By contradiction, suppose $H_{\infty}>0$. Then $\tfrac{1}{2}y^{2}%
+\rho\rightarrow3H_{\infty}^{2}$ and therefore $\sup\mathbb{I}\in\mathbb{R}$,
otherwise it would be
\[
H(t)-H(0)=-\frac{1}{2}\int_{t_{0}}^{t}y(s)^{2}+\gamma\rho(s)\,\mathrm{d}%
s\rightarrow-\infty,
\]
as $t\rightarrow\infty$, a contradiction. By the Cauchy-Schwarz inequality we
obtain,
\begin{align*}
(\phi(t)-\phi(t_{0}))^{2}  &  =\left(  \int_{t_{0}}^{t}\dot{\phi
}(s)\,\mathrm{d}s\right)  ^{2}\leq(t-t_{0})\int_{t_{0}}^{t}y(s)^{2}%
\,\mathrm{d}s\\
&  \leq(t-t_{0})\int_{t_{0}}^{t}-2\dot{H}(s)\,\mathrm{d}s=2(t-t_{0}%
)(H(t_{0})-H(t)),
\end{align*}
that converges to the finite value $2(\sup\mathbb{I}-t_{0})(H(t_{0}%
)-H_{\infty})\in\mathbb{R}$, hence $\phi(t)$ is bounded, which is a
contradiction. Therefore, $H_{\infty}\leq0$.
\end{proof}

\begin{proposition}
\label{thm:propV0} Suppose that $\phi_{\infty}=+\infty\ $and $V_{\infty}=0$.
If \eqref{eq:cond} does not hold, then \eqref{eq:rec} generically holds.
Otherwise, i.e. if \eqref{eq:cond} holds, either \eqref{eq:rec} generically
holds or the solution expands forever, with $\phi(t)\rightarrow+\infty$ and
$y(t),\rho(t)$ and $H(t)$ infinitesimal as $t\rightarrow+\infty$.
\end{proposition}

\begin{proof}
By the preceding Lemma, $H_{\infty}\leq0$. If $H_{\infty}$ is strictly
negative then the results follows from Lemma \ref{lem:2}. Suppose now it is
zero; this means that the solution expands forever and a normalized variables
scheme can be used to study the critical point\ \textquotedblleft at
infinity\textquotedblright. We use variables $(\phi,x,w,z)$ as in Proposition
\ref{thm:prop2}, which are functions of a new time $\tau$ coordinate defined
by $\mathrm{d}\tau/\mathrm{d}t=H$\textrm{,} (note that unlike the case treated
in Proposition \ref{thm:prop2}, now $H>0$). Using the function $g$ defined in
Assumption \ref{as:1} and arguing as in the proof of Proposition
\ref{thm:prop2}, we arrive at the following system,
\begin{align}
\frac{\mathrm{d}w}{\mathrm{d}\tau}  &  =\left(  \frac{1}{2}w^{2}-3\right)
(w+u(g^{-1}(s)))+z^{2}\left(  \frac{\gamma}{2}w+\alpha(g^{-1}(s))+u(g^{-1}%
(s))\right)  ,\label{eq:A6}\\
\frac{\mathrm{d}z}{\mathrm{d}\tau}  &  =\frac{1}{2}z\left[  w^{2}%
-\alpha(g^{-1}(s))w+\gamma(z^{2}-3)\right]  ,\label{eq:A7}\\
\frac{\mathrm{d}s}{\mathrm{d}\tau}  &  =wg^{\prime}(g^{-1}(s)). \label{eq:A8}%
\end{align}
We are interested in solutions such that $s\rightarrow0$ and $\frac{1}{2}%
w^{2}+z^{2}>3$, with $w,z\geq0$ eventually. Therefore, the $(w,z)$%
--coordinates of the critical points that are admissible candidates to be
$\omega$--limit points are, (setting $\lambda=\lambda_{+}$ and $\alpha
=\alpha(g^{-1}(0))=\tfrac{4-3\gamma}{2}Q_{+}$)
\begin{align*}
\mathcal{A}  &  =\left(  \sqrt{6},0\right)  ,\,\mathcal{B}=(-\lambda
,0),\,\mathcal{C}=\left(  \frac{2\alpha}{2-\gamma},\frac{\sqrt{-2\alpha
^{2}+3(2-\gamma)^{2}}}{2-\gamma}\right)  ,\,\\
\mathcal{D}  &  =\left(  -\frac{3\gamma}{\alpha+\lambda},-\frac{\sqrt
{3(-3\gamma+\alpha\lambda+\lambda^{2})}}{\alpha+\lambda}\right)  .
\end{align*}
The analysis of these critical points reveals that the only sink can be
$\mathcal{C}$, and this happens precisely when \eqref{eq:cond} holds. In this
case, we obtain ever expanding solutions such that $H_{\infty}=0$, and
consequently, both $y$ and $\rho$ tend to zero; since the solution is defined
for $\tau\rightarrow+\infty$, and recalling that $t$ is an increasing function
of $\tau$, we get $\sup\mathbb{I}=+\infty$, i.e. also the corresponding
solution to \eqref{eq:1}--\eqref{eq:5} is defined for $t\rightarrow+\infty$.

If the solution does not start into the basin of attraction of $\mathcal{C}$,
it is unbounded, thus, $\frac{1}{2}w^{2}+z^{2}\rightarrow+\infty$. Suppose by
contradiction that they also correspond to ever expanding cosmologies with
$H_{\infty}=0$. Then $\sup\mathbb{I}=+\infty$. Set
\[
\widetilde{x}=\frac{\sqrt{|V|}}{H},
\]
and observe that
\[
\frac{\mathrm{d}\widetilde{x}}{\mathrm{d}\tau}=\frac{1}{2}\widetilde
{x}(u(g^{-1}(s))w+w^{2}+\gamma z^{2})\approx\frac{1}{2}\widetilde{x}%
(w^{2}+\gamma z^{2}),
\]
where the symbol $\approx$ is used to denote the dominant terms. Now,
\[
\widetilde{x}^{2}=\frac{1}{2}w^{2}+z^{2}-3\approx\frac{1}{2}w^{2}+z^{2}\approx
K(w^{2}+\gamma z^{2}),
\]
for some constant $K>0$, so $\mathrm{d}\widetilde{x}/\mathrm{d}\tau\approx
A\widetilde{x}^{3}$ for some $A>0$, which implies $\widetilde{x}(\tau
)\approx(a-b\tau)^{-1/2}$ for suitable $a,b>0$.

Then,
\[
\frac{1}{2}(w^{2}+\gamma z^{2})\approx\frac{\mathrm{d}\widetilde{x}%
}{\mathrm{d}\tau}\frac{1}{\widetilde{x}}\approx\frac{b}{2}(a-b\tau)^{-1},
\]
hence,
\[
\frac{\mathrm{d}H}{\mathrm{d}\tau}=-H\frac{1}{2}(w^{2}+\gamma z^{2}%
)\approx-H(\tau)\frac{b}{2}(a-b\tau)^{-1},
\]
from which $H(\tau)\approx H_{0}\sqrt{a-b\tau}$. This implies that
\[
t=\int_{\tau_{0}}^{\tau}\frac{1}{H(\sigma)}\mathrm{d}\sigma\approx\int
_{\tau_{0}}^{\tau}\frac{1}{H_{0}\sqrt{a-b\sigma}}\,\mathrm{d}\sigma,
\]
which converges as $\tau\rightarrow a/b$. This means that $\sup\mathbb{I}%
\in\mathbb{R}$, that is a contradiction. Therefore, $H(t)<0$ eventually. Since
$V(\phi(t))<0$ eventually, the conclusion follows from Lemma \ref{lem:2}.
\end{proof}

\begin{example}
\label{rem:example} To illustrate the situation depicted in Proposition
\ref{thm:propV0}, let us consider the double exponential potential,
\[
V(\phi)=e^{\lambda_{-}\phi}-e^{\lambda_{+}\phi},
\]
that falls into case \ref{itm:A}, provided that $\lambda_{-}<\lambda_{+}%
<0$.\begin{figure}[th]
\begin{center}
\subfigure[\ ]{            \label{fig:first}
\includegraphics[width=0.4\textwidth]{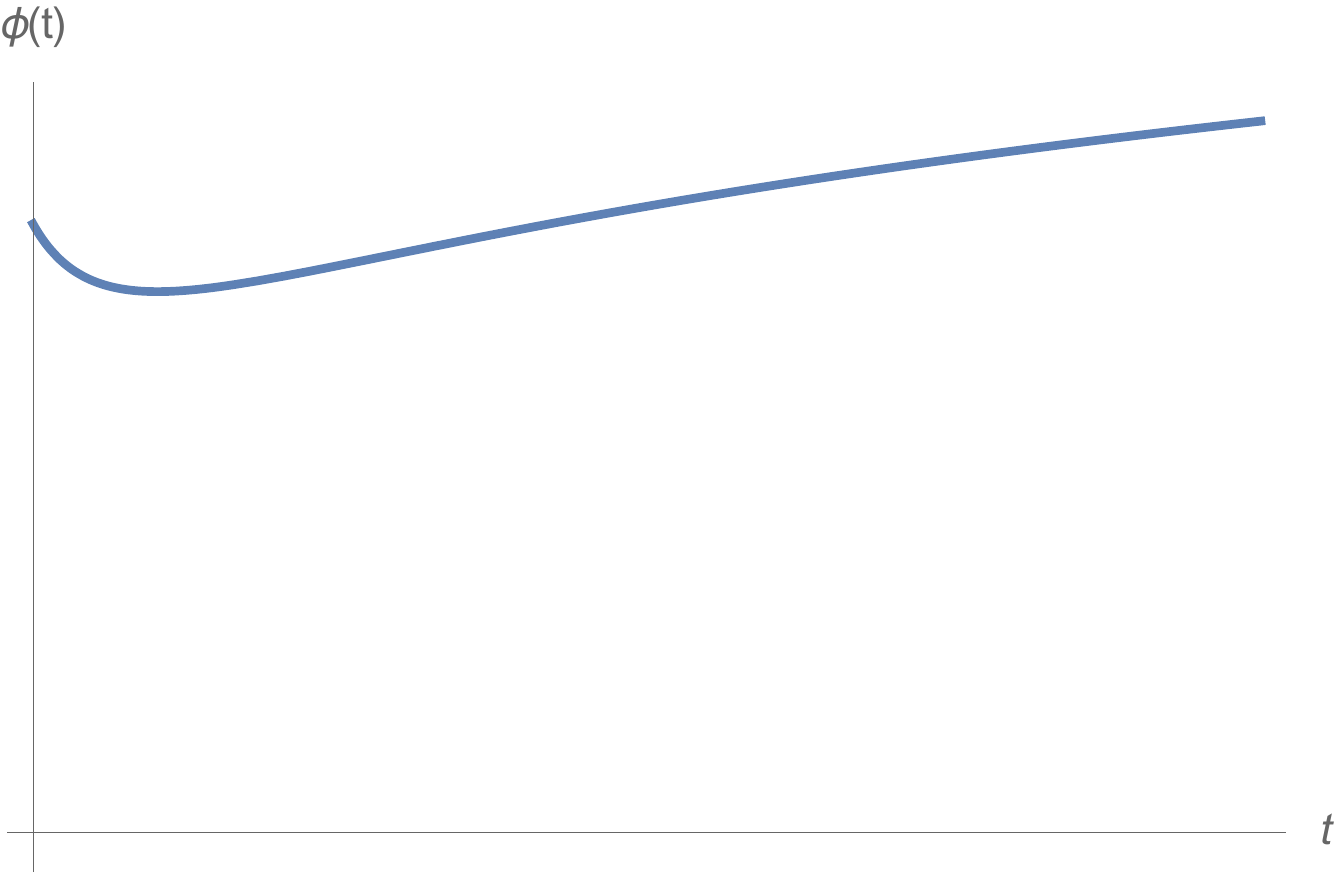}
} \subfigure[\ ]{           \label{fig:second}
\includegraphics[width=0.4\textwidth]{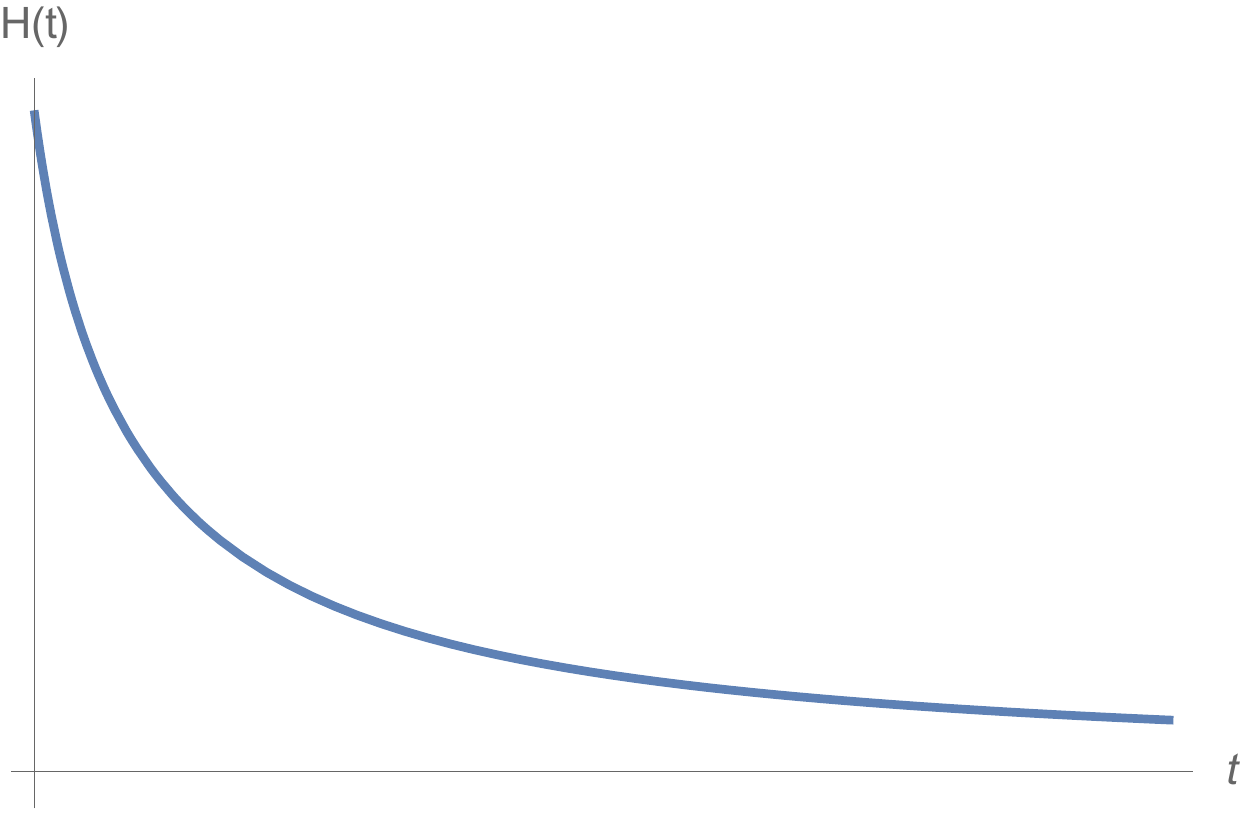}
}\newline\subfigure[\ ]{            \label{fig:third}
\includegraphics[width=0.4\textwidth]{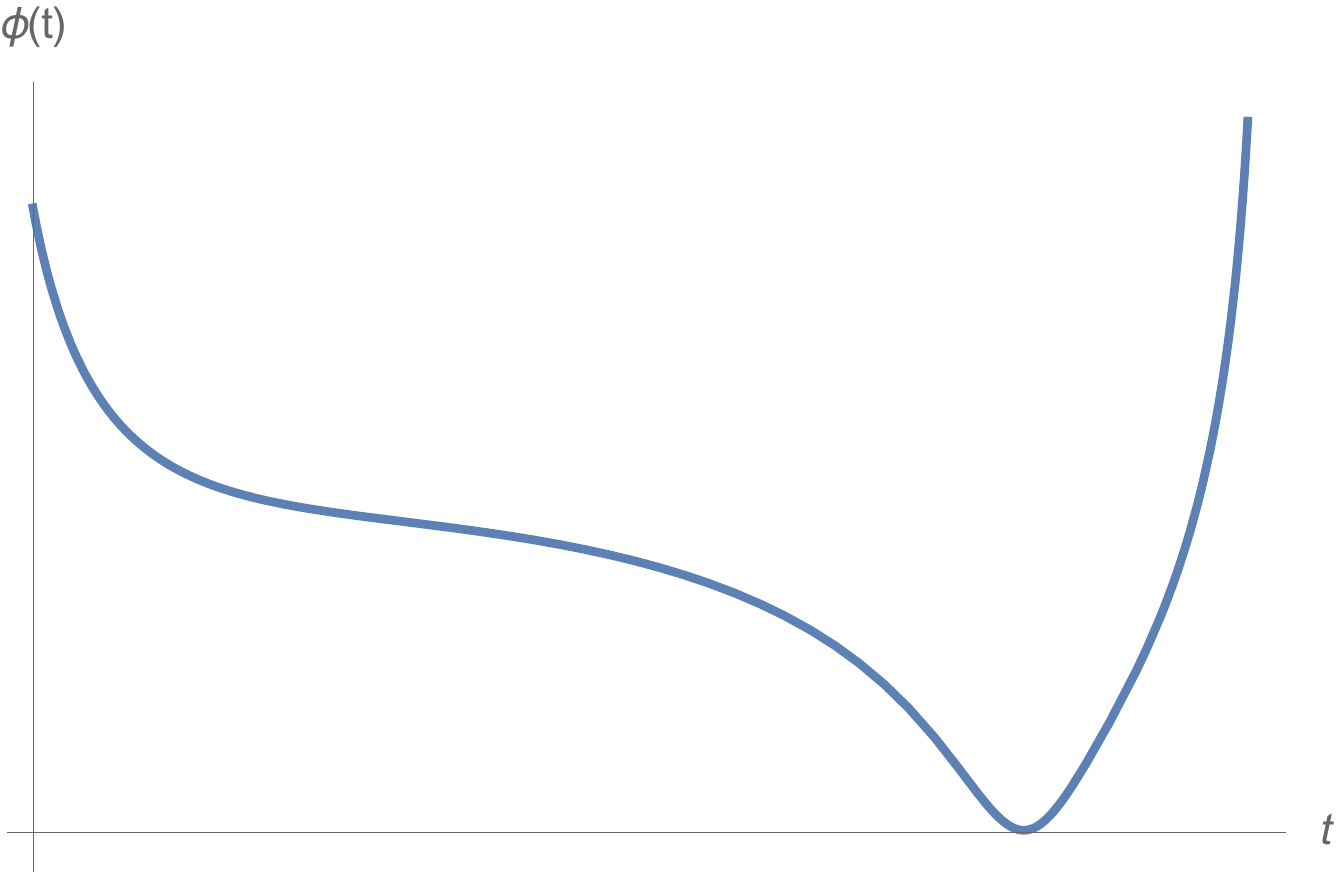}
} \subfigure[\ ]{            \label{fig:fourth}
\includegraphics[width=0.4\textwidth]{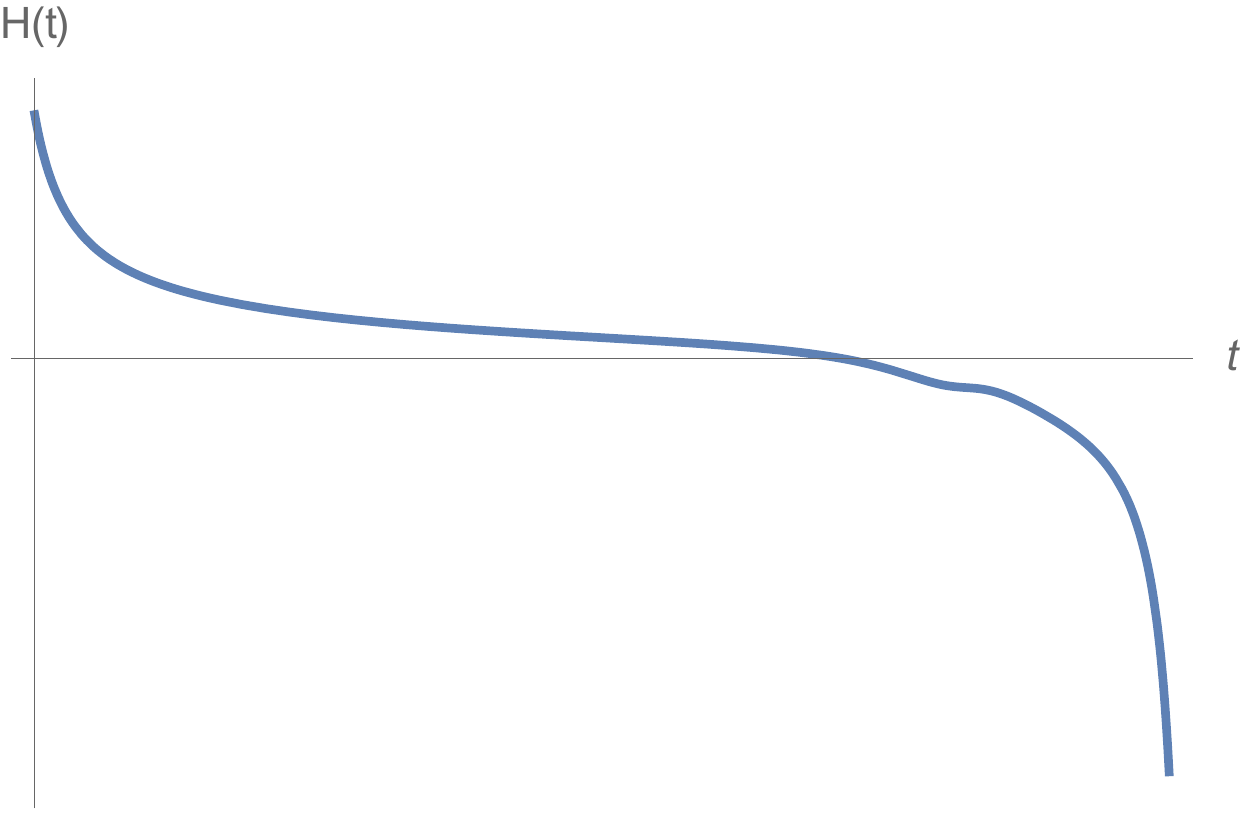}
}
\end{center}
\caption{The cases studied in Proposition \ref{thm:propV0}. In the top figures
\ref{fig:first}--\ref{fig:second}, the scalar field positively diverges and
the Universe expands forever. I the bottom figures \ref{fig:third}%
--\ref{fig:fourth}, the scalar field positively diverges again, but in a
finite time, resulting in recollapse and development of a singularity.}%
\label{fig:subfigures}%
\end{figure}

Choose parameters, for example, $\lambda_{+}=-4$, $\lambda_{-}=-5$, and
$Q(\phi)=1$, a constant. Then condition \eqref{eq:cond} is satisfied for a non
empty set of the admissible values of $\gamma$, for instance $\gamma=1$,
(dust). With this set of parameters, both expansion and recollapse may take
place, depending on the initial conditions.

With initial conditions for instance, $H(0)=1$, $\phi(0)=2$ and $y(0)=-1$ (the
initial value $\rho(0)$ is not arbitrary, but is determined by \eqref{eq:5}),
the scalar field positively diverges in an infinite time and the Hubble
function remains always positive, tending asymptotically to zero; therefore
the Universe expands forever, Figures \ref{fig:first}--\ref{fig:second}.

Simply changing the initial conditions, for instance $y(0)=-2$, then the
scalar positively diverges again, but now in a finite amount of time. Indeed,
$H(t)$ changes sign and once it becomes negative, the solution is forced to
recollapse and develop a singularity, Figures \ref{fig:third}%
--\ref{fig:fourth}.
\end{example}

\subsection{Case $\phi_{\infty}$ does not exist}

In this subsection we study the case when $\phi(t)$ neither converges nor diverges.

\begin{proposition}
\label{thm:phinotexists} If $\phi_{\infty}$ does not exist, then
\eqref{eq:rec} generically holds.
\end{proposition}

\begin{proof}
First we claim that
\begin{equation}
H_{\infty}=\lim_{t\rightarrow\sup\mathbb{I}}H(t)=-\infty, \label{eq:hinfty}%
\end{equation}
generically holds, by considering the following subcases.

\begin{enumerate}
\item Suppose $\liminf_{t\rightarrow\sup\mathbb{I}}V(\phi(t))\geq0$. If by
contradiction, $H(t)$ was bounded, then from \eqref{eq:5} we could conclude
that $y(t),\rho(t)$ were bounded too. Then, for the solution to be  generic
(recall again Lemma \ref{lem:1}), $\phi(t)$ should be unbounded. But since
$V(\phi(t))$ must be eventually non negative, this would imply that
$\limsup_{t\rightarrow\sup\mathbb{I}}V(\phi(t))=+\infty$. Then a sequence
$t_{n}\rightarrow\sup\mathbb{I}$ exists, such that $H(t_{n})^{2}%
\rightarrow+\infty$, which means that $H(t_{n})\rightarrow-\infty$, which is a
contradiction. Thus, $H(t)$ cannot be bounded and therefore \eqref{eq:hinfty}
must hold.

\item Suppose $\liminf_{t\rightarrow\sup\mathbb{I}}V(\phi(t))<0$. In this case
there exist sequences $\{t_{n}\},\{s_{n}\}$, such that
\[
t_{n},s_{n}\rightarrow\sup\mathbb{I},\quad t_{n}<s_{n}<t_{n+1},
\]
with $V(\phi\left(  t_{n}\right)  ),V(\phi\left(  s_{n}\right)  )<0$ and
$\phi(t)$ lies between $\phi\left(  t_{n}\right)  $ and $\phi\left(
s_{n}\right)  $, $\forall t\in\lbrack t_{n},s_{n}]$. Using Cauchy-Schwarz
inequality as in the proof of Lemma \ref{thm:lem} we get
\[
(\phi\left(  t_{n}\right)  -\phi\left(  s_{n}\right)  )^{2}\leq(s_{n}%
-t_{n})\int_{t_{n}}^{s_{n}}-2\dot{H}(s)\,\mathrm{d}s=2(s_{n}-t_{n}%
)(H(t_{n})-H(s_{n})),
\]
and therefore
\begin{equation}
s_{n}-t_{n}\geq\frac{(\phi\left(  t_{n}\right)  -\phi\left(  s_{n}\right)
)^{2}}{2(H(t_{n})-H(s_{n}))}. \label{eq:deltat}%
\end{equation}
Now, if by contradiction $H_{\infty}\in\mathbb{R}$ then \eqref{eq:deltat}
would imply that $s_{n}-t_{n}\rightarrow+\infty$ and as a consequence
$\sup\mathbb{I}=+\infty$. Moreover comparison theorems in ODE would say that
$H(t)\leq z(t)$ in $[t_{n},s_{n}]$, where $z(t)$ solves the Cauchy problem
\[
\dot{z}(t)=\frac{\gamma}{2}(-3z(t)^{2}+\bar{V}),\qquad z(t_{n})=H(t_{n}),
\]
and $\bar{V}$ is a (negative) constant such that $V(\phi)<\bar{V}$, for every
$\phi$ between $\phi_{t}$ and $\phi_{s}$. Now, observe that the solution
$z(t)$ to the Cauchy problem above negatively diverges for some $t_{n}%
+\delta_{n}$, where $\delta_{n}$ is uniformly bounded with respect to $n$,
whereas $s_{n}-t_{n}\rightarrow+\infty$, and this is a contradiction. Hence
$H_{\infty}=-\infty$, i.e. \eqref{eq:hinfty} holds.
\end{enumerate}

In both cases (2a) and (2b) we have shown that \eqref{eq:hinfty} holds. Let us
prove that this happens in a finite amount of time. If $\liminf_{t\rightarrow
\sup\mathbb{I}}\phi(t)\in\mathbb{R}$ then there exists a $\bar{V}\in
\mathbb{R}$, such that $V(\phi(t))\leq\bar{V}$ eventually, and the result
follows from Lemma \ref{lem:2}.

If $\liminf_{t\rightarrow\sup\mathbb{I}}\phi(t)=-\infty$, i.e., $\limsup
_{t\rightarrow\sup\mathbb{I}}V(\phi(t))=+\infty$, then we can consider the
same system in normalized variables used in case (1c) before, see Remark
\ref{rem:Vdiv} after Proposition \ref{thm:prop2}.
\end{proof}

\section{Potentials of class \ref{itm:D}}

\label{sec4} The treatment of potentials of class \ref{itm:D} is similar to
the methods used in \cite[Theorem 1]{gmt}. In fact, we have the following result.

\begin{proposition}
\label{thm:D} \label{thm:prop1} Let $V(\phi)\in\mathcal{C}^{2}$ such that

\begin{enumerate}
\item \label{itm:p1} $\lim_{\phi\to\pm\infty} V(\phi)=-\infty$,

\item \label{itm:p2} $V$ has a unique nondegenerate critical point (that has
to be, in view of \eqref{itm:p1}, the global maximum),

\item \label{itm:p3} There exist $\lambda>0$ and $M>0$ such that, $|V^{\prime
}(\phi)|\leq-\lambda V(\phi)$, for all $\phi\,:\,|\phi|>M$.
\end{enumerate}

Then a solution to \eqref{eq:1}--\eqref{eq:5} generically recollapses in a
finite time, i.e. (\ref{eq:rec}) holds.
\end{proposition}

The argument follows the same line of the proof used for left unbounded
potentials treated in \cite{gmt}. Indeed, in that paper a class of potentials
was considered such that hypotheses \eqref{itm:p1} and \eqref{itm:p3} hold
only for $\phi\rightarrow-\infty$, whereas $V(\phi)\rightarrow0^{+}$ as
$\phi\rightarrow+\infty$. A critical value $\dot{\phi}_{crit}$ is shown to
exist such that, if $y(t_{0})<\dot{\phi}_{crit}$, the scalar field is
eventually forced to take values to the left of the global maximum. The proof
of Theorem 1 in \cite{gmt}, which is the analogue of the above Proposition
\ref{thm:D}, relies precisely on these hypotheses, in such a way that one can
be sure that the scalar field does not positively diverge to the critical
point at infinity, a case examined in a previous paper \cite{tzmi}. In the
present case, potentials of class \ref{itm:D} do not have that complication,
diverging to $-\infty$ on both directions, and therefore, regardless of the
behavior of the scalar field, and recalling Lemma \ref{lem:1}, it can be
proved with exactly the same argument that, the solution recollapses almost
always and the Hubble function negatively diverges in a finite time.

\section{Discussion\label{end}}

With Proposition \ref{thm:prop1} we complete the analysis, carried on in Ref.
\cite{gmt}, of the class of potentials falling to minus infinity as
$\phi\rightarrow-\infty$, having a global positive maximum and going to zero
from above as $\phi\rightarrow+\infty$. In that study, assuming that the
growth of $|V(\phi)|$ to infinity is at most exponential, the corresponding
initially expanding Universes, eventually collapse in a finite time, up to a
set of initial data of measure zero. In the present paper we extend the
analysis to situations where the negative branch of the potential function may
possibly not diverge, whereas the positive branch diverges to $+\infty$. We
have investigated the qualitative behavior of the Hubble function, examining
all possible cases for the asymptotic behavior of the scalar field. We have
found that the recollapse and the formation of a future singularity always
take place in a \textit{generic} way, i.e. stable with respect to
perturbations of the initial data of the system. Moreover, recollapse is the
only generical situation allowed, except in case the potential goes to zero
from below as $\phi\rightarrow+\infty$ and \eqref{eq:cond} holds; in this case
there also exists generical choices of initial data that do not lead to
recollapse, producing an ever--expanding cosmology where the scalar field
positively diverges. Our conclusions are valid for scalar fields coupled to
matter, as well as for uncoupled models studied so far in the literature.

Cosmology with negative potentials is the basis of the cyclic Universes in the
context of the ekpyrotic scenario. Our results may be helpful in building
solid models of cyclic cosmologies and therefore avoid the fragility of this
scenario with respect to the unknown physics at the singularity.

\end{document}